\documentclass[USenglish,a4paper,11pt]{article}
\pdfoutput=1

\usepackage[T1]{fontenc}
\usepackage[margin=1in]{geometry}
\usepackage{authblk}
\usepackage{microtype}
\microtypecontext{spacing=nonfrench}
\usepackage[backref=true,style=alphabetic,citestyle=alphabetic,minalphanames=3,maxalphanames=4,maxbibnames=99,maxcitenames=99]{biblatex}
\usepackage{amsmath, amsthm, amssymb}
\usepackage{graphicx}
\usepackage{mathtools}
\usepackage{enumitem}
\usepackage[dvipsnames]{xcolor} 
\usepackage{booktabs}
\usepackage{bbm} 
\usepackage{tcolorbox} 
\usepackage{makecell}
\usepackage{mathrsfs}

\DeclareOldFontCommand{\sc}{\normalfont\scshape}{\mathsc}
\DeclareOldFontCommand{\tt}{\normalfont\ttfamily}{\mathtt}
\DeclareOldFontCommand{\it}{\normalfont\itshape}{\mathit}
\DeclareOldFontCommand{\bf}{\normalfont\bfseries}{\mathbf}

\usepackage{hyperref}
\hypersetup{
    pdftitle={Quantum Sub-Gaussian Mean Estimator},
    pdfauthor={Yassine Hamoudi},
    plainpages=false, 
    colorlinks=true, linkcolor=BrickRed, citecolor=MidnightBlue, filecolor=blue, urlcolor=OliveGreen}
\usepackage{bookmark} 

\theoremstyle{plain}
\newtheorem{theorem}{Theorem}[section]
\newtheorem{proposition}[theorem]{Proposition}
\newtheorem{corollary}[theorem]{Corollary}
\newtheorem{lemma}[theorem]{Lemma}

\theoremstyle{definition}
\newtheorem{definition}[theorem]{Definition}
\newtheorem{assumption}{Assumption}
\newenvironment{asspt}[1]{%
  \renewcommand\theassumption{#1}
  \assumption
}{\endassumption}

\newtheoremstyle{restate}{}{}{\itshape}{}{\bfseries}{.}{.5em}{\thmnote{#3}}
\theoremstyle{restate}
\newtheorem*{rtheorem}{Theorem}

\newlist{enumresult}{enumerate}{1}
\setlist[enumresult,1]{label=(\arabic*), font=\normalfont\ttfamily}

\newcommand{\R}{\mathbb{R}}

\newcommand{\C}{\mathbb{C}}

\newcommand{\Hil}{\mathcal{H}}
\newcommand{\ra}{\rightarrow}

\newcommand{\rn}{\{0,1\}}
\newcommand{\eps}{\epsilon}
\newcommand\super[1]{^{(#1)}}
\newcommand\td[1]{\widetilde{#1}}

\newcommand\ha[1]{\widehat{#1}}
\newcommand\ind[1]{\mathbbm{1}_{#1}}
\newcommand{\qub}{\ket{\mathbf{0}}}   

\DeclareMathOperator{\median}{median}

\let\ket\relax
\let\braket\relax
\let\proj\relax
\providecommand\given{}
\DeclarePairedDelimiterX{\pt}[1](){\renewcommand\given{\nonscript\:\delimsize\vert\nonscript\:\mathopen{}}#1} 
\DeclarePairedDelimiterX{\bc}[1][]{\renewcommand\given{\nonscript\:\delimsize\vert\nonscript\:\mathopen{}}#1} 
\DeclarePairedDelimiterX{\bcr}[1][){\renewcommand\given{\nonscript\:\delimsize\vert\nonscript\:\mathopen{}}#1} 
\DeclarePairedDelimiter{\ceil}{\lceil}{\rceil}

\DeclarePairedDelimiter{\set}{\{}{\}}              
\DeclarePairedDelimiter{\abs}{\lvert}{\rvert}      
\DeclarePairedDelimiter{\norm}{\lVert}{\rVert}     

\DeclarePairedDelimiter{\ket}{|}{\rangle}
\DeclarePairedDelimiterX{\inp}[2]{\langle}{\rangle}{#1, #2}
\DeclarePairedDelimiterX{\ip}[2]{\langle}{\rangle}{#1\,\delimsize\vert\,\mathopen{}#2}
\DeclarePairedDelimiterX{\op}[2]{|}{|}{#1 \delimsize\rangle \delimsize\langle \mathopen{}#2}
\DeclarePairedDelimiterX{\braket}[3]{\langle}{\rangle}{#1\,\delimsize\vert\,\mathopen{}#2\,\delimsize\vert\,\mathopen{}#3}
\DeclarePairedDelimiterX{\proj}[1]{|}{|}{#1 \delimsize\rangle \delimsize\langle \mathopen{}#1}

\newcommand{\bo}{O\pt}
\newcommand{\wbo}{\widetilde{O}\pt}

\newcommand{\om}{\Omega\pt}
\newcommand{\wom}{\widetilde{\Omega}\pt}
\newcommand{\ta}{\Theta\pt}
\newcommand{\wta}{\widetilde{\Theta}\pt}

\newcommand{\Ex}{\mathbb{E}}
\newcommand{\Var}{\mathrm{Var}}
\newcommand{\ex}{\Ex\bc}
\newcommand{\var}{\Var\bc}
\newcommand{\pr}{\Pr\bc}

\newcommand{\proc}[1]{\textup{\textsf{#1}}} 

\newcommand{\mut}{\td{\mu}}
\newcommand{\muh}{\ha{\mu}}

\newcommand{\aamp}{\proc{AAmp}}
\newcommand{\aest}{\proc{AEst}}
\newcommand{\saamp}{\proc{Seq-AAmp}}
\newcommand{\saest}{\proc{Seq-AEst}}

\newcommand{\qrv}{q-random variable}
\newcommand{\qrvs}{q-random variables}
\newcommand{\subgauss}{\proc{SubGaussEst}}
\newcommand{\bern}{\proc{BernEst}}
\newcommand{\quant}{\proc{Quantile}}
\newcommand{\sbern}{\proc{Seq-BernEst}}

\newcommand{\ch}{\Delta} 
\newcommand{\nb}{n}
\newcommand{\mb}{m}

\newcommand{\algobox}[3]{
\renewcommand{\figurename}{Algorithm}
  \begin{figure}[htb] 
    \centering
    \begin{tcolorbox}
      #3
    \end{tcolorbox}
    \caption{#2}
    \label{#1}
  \end{figure}
\renewcommand{\figurename}{Figure}}

\title{Quantum Sub-Gaussian Mean Estimator}
\author{Yassine Hamoudi}
\affil{Universit\'e de Paris, IRIF, CNRS, F-75006 Paris, France.\\\texttt{hamoudi@irif.fr}}
\date{\today}

\begin{document}

\maketitle


\begin{abstract}
  We present a new quantum algorithm for estimating the mean of a real-valued random variable obtained as the output of a quantum computation. Our estimator achieves a nearly-optimal quadratic speedup over the number of classical i.i.d. samples needed to estimate the mean of a heavy-tailed distribution with a sub-Gaussian error rate. This result subsumes (up to logarithmic factors) earlier works on the mean estimation problem that were not optimal for heavy-tailed distributions~\cite{BHMT02j,BDGT11p}, or that require prior information on the variance~\cite{Hei02j,Mon15j,HM19c}. As an application, we obtain new quantum algorithms for the $(\eps,\delta)$-approximation problem with an optimal dependence on the coefficient of variation of the input random variable.

\end{abstract}



\section{Introduction}

The problem of estimating the mean $\mu$ of a real-valued random variable $X$ given \emph{i.i.d.} samples from it is one of the most basic tasks in statistics and in the Monte Carlo method. The properties of the various classical \emph{mean estimators} are well understood. The standard non-asymptotic criterion used to assess the quality of an estimator is formulated as the following \emph{high probability deviation bound}: upon performing $\nb$ random experiments that return $\nb$ samples from $X$, and given a failure probability $\delta \in (0,1)$, what is the smallest error $\eps(\nb,\delta,X)$ such that the output~$\mut$ of the estimator satisfies $\abs{\mut - \mu} > \eps(\nb,\delta,X)$ with probability at most $\delta$? Under the standard assumption that the unknown random variable~$X$ has a finite variance $\sigma^2$, the best possible performances are obtained by the so-called \emph{sub-Gaussian estimators}~\cite{LM19j} that achieve the following deviation bound
  \begin{equation}
    \label{Eq:Gaussian}
    \pr*{\abs{\mut - \mu} > L \sqrt{\frac{\sigma^2 \log(1/\delta)}{\nb}}} \leq \delta
  \end{equation}
for some constant $L$. The term ``sub-Gaussian'' reflects that these estimators have a Gaussian tail even for non-Gaussian distributions. The most well-known sub-Gaussian estimator is arguably the \emph{median-of-means}~\cite{NY83b,JVV86j,AMS99j}, which consists of partitioning the $\nb$ samples into roughly $\log(1/\delta)$ groups of equal size, computing the empirical mean over each group, and returning the median of the obtained means.

The process of generating a random sample from $X$ is generalized in the quantum model by assuming the existence of a unitary operator $U$ where $U \qub$ coherently encodes the distribution of~$X$. A \emph{quantum experiment} is then defined as one application of this operator or its inverse. The celebrated quantum amplitude estimation algorithm~\cite{BHMT02j} provides a way to estimate the mean of any \emph{Bernoulli} random variable by performing fewer experiments than with any classical estimator. Yet, for general distributions, the existing quantum mean estimators either require additional information on the variance~\cite{Hei02j,Mon15j,HM19c} or are less performant than the classical sub-Gaussian estimators when the distribution is heavy tailed~\cite{BHMT02j,Ter99d,BDGT11p,Mon15j}. These results leave open the existence of a general quantum speedup for the mean estimation problem. We address this question by introducing the concept of \emph{quantum sub-Gaussian estimators}, defined through the following deviation bound
  \begin{equation}
    \label{Eq:QGaussian}
    \pr*{\abs{\mut - \mu} > L \frac{\sigma \log(1/\delta)}{\nb}} \leq \delta
  \end{equation}
for some constant $L$. We give the first construction of a quantum estimator that achieves this bound up to a logarithmic factor in $\nb$. Additionally, we prove that it is impossible to go below that deviation level. This result provides a clear equivalent of the concept of sub-Gaussian estimator in the quantum setting.

A second important family of mean estimators addresses the \emph{$(\eps,\delta)$-approximation} problem, where given a fixed relative error $\eps \in (0,1)$ and a failure probability $\delta \in (0,1)$ the goal is to output a mean estimate $\mut$ such that
  \begin{equation}
    \label{Eq:EpsDelta}
    \pr*{\abs{\mut - \mu} > \eps \abs{\mu}} \leq \delta.
  \end{equation}
The aforementioned sub-Gaussian estimators do not quite answer this question since the number of experiments they require (respectively $\nb = \om[\big]{\pt{\frac{\sigma}{\eps \mu}}^2 \log(1/\delta)}$ and $\nb = \wom[\big]{\frac{\sigma}{\eps \abs{\mu}} \log(1/\delta)}$) depends on the \emph{unknown} quantities $\sigma$ and $\mu$. Sometimes a good upper bound is known on the \emph{coefficient of variation} $\abs{\sigma/\mu}$ and can be used to parametrize a sub-Gaussian estimator. Otherwise, the standard approach is based on \emph{sequential analysis} techniques, where the number of experiments is chosen adaptively depending on the results of previous computations. Given a random variable distributed in $[0,1]$, the optimal classical estimators perform $\ta[\big]{\pt[\big]{\pt[\big]{\frac{\sigma}{\eps \mu}}^2 + \frac{1}{\eps \mu} }\log(1/\delta)}$ random experiments \emph{in expectation}~\cite{DKLR00j} for computing an $(\eps,\delta)$-approximation of $\mu$. We construct a quantum estimator that reduces this number to  $\wta[\big]{\pt[\big]{\frac{\sigma}{\eps \mu} + \frac{1}{\sqrt{\eps \mu}}}\log(1/\delta)}$ and we prove that it is optimal.


\subsection{Related work}

There is an extensive literature on classical sub-Gaussian estimators and we refer the reader to \cite{LM19j,Cat12j,BCL13j,DLLO16j,LV20p} for an overview of the main results and recent improvements. We point out that the \emph{empirical mean} estimator is not sub-Gaussian, although it is optimal for Gaussian random variables~\cite{SV05j,Cat12j}. The non-asymptotic performances of the empirical mean estimator are captured by several standard concentration bounds such as the Chebyshev, Chernoff and Bernstein inequalities.

There is a series of quantum mean estimators \cite{Gro98c,AW99p,BDGT11p} that get close to the bound $\pr[\Big]{\abs{\mut - \mu} > L \frac{\log(1/\delta)}{\nb}} \leq \delta$ for any random variable distributed in $[0,1]$ and some constant $L$. Similar results hold for numerical integration problems \cite{AW99p,Nov01j,Hei02j,TW02j,Hei03j}. The amplitude estimation algorithm~\cite{BHMT02j,Ter99d} leads to a sharper bound of $\pr[\Big]{\abs{\mut - \mu} > L \pt[\Big]{\frac{\sqrt{\mu(1-\mu)}\log(1/\delta)}{\nb} + \frac{\log(1/\delta)^2}{\nb^2}}} \leq \delta$ (see Proposition~\ref{Prop:ZeroOne}) when $X$ is distributed in $[0,1]$. Nevertheless, the quantity $\mu(1-\mu)$ is always larger than or equal to the variance $\sigma^2$. The question of improving the dependence on $\sigma^2$ was considered in~\cite{Hei02j,Mon15j,HM19c}. The estimators of~\cite{Hei02j,Mon15j} require to know an upper bound $\Sigma$ on the standard deviation $\sigma$, whereas \cite{HM19c} needs an upper bound $\ch$ on the coefficient of variation $\sigma/\mu$ (for non-negative random variables). The performances of these estimators are captured (up to logarithmic factors) by the deviation bound given in Equation~(\ref{Eq:QGaussian}) with $\sigma$ replaced by $\Sigma$ and $\mu\ch$ respectively.

The $(\eps,\delta)$-approximation problem has been addressed by several classical works such as \cite{DKLR00j,MSA08c,GNP13j,Hub19j}. In the quantum setting, there is a variant~\cite[Theorem 15]{BHMT02j} of the amplitude estimation algorithm that performs $\bo{\log(1/\delta)/(\eps \sqrt{\mu})}$ experiments in expectation to compute an $(\eps,\delta)$-approximate of the mean of a random variable distributed in $[0,1]$ (see Theorem~\ref{Thm:SeqAE} and Proposition~\ref{Prop:SeqZeroOne}). However, the complexity of this estimator does not scale with $\sigma$. Given an upper bound $\ch$ on $\sigma/\mu$, the estimator of~\cite{HM19c} can be used to compute an $(\eps,\delta)$-approximate with roughly $\wbo{\ch \log(1/\delta)/\eps}$ quantum experiments if the random variable is non-negative.

We note that the related problem of estimating the mean with \emph{additive} error $\eps$, that is $\pr{\abs{\mut - \mu} > \eps} \leq \delta$, has also been considered by several authors. The optimal number of experiments is $\ta{\log(1/\delta)/\eps^2}$ classically~\cite{CEG95j} and $\ta{1/\eps}$ quantumly~\cite{NW99c} (with failure probability $\delta = 1/3$). These bounds do not depend on unknown parameters (as opposed to the relative error case), thus sequential analysis techniques are unnecessary here. Montanaro~\cite{Mon15j} also described an estimator that performs $\wbo{\Sigma \log(1/\delta)/\eps}$ quantum experiments given an upper bound $\Sigma$ on the standard deviation $\sigma$.


\subsection{Contributions and organization}

We first formally define the input model in Section~\ref{Sec:ModelSampling}. We introduce the concept of ``\qrv'' (Definition~\ref{Def:QVar}) to describe a random variable that corresponds to the output of a quantum computation. We measure the complexity of an algorithm by counting the number of \emph{quantum experiments} (Definition~\ref{Def:qExp}) it performs with respect to a \qrv.

We construct a quantum algorithm for estimating the quantiles of a \qrv\ in Section~\ref{Sec:Quantile}, and we use it in Section~\ref{Sec:SubGaussian} to design the following quantum sub-Gaussian estimator.

\begin{rtheorem}[Theorem~\ref{Thm:SubGaussian} {\normalfont (Restated)}]
  There exists a quantum algorithm with the following properties.
  Let~$X$ be a \qrv\ with mean $\mu$ and variance $\sigma^2$, and set as input a time parameter~$\nb$ and a real $\delta \in (0,1)$ such that $\nb \geq \log(1/\delta)$. Then, the algorithm outputs a mean estimate $\mut$ such that,
    $\pr*{\abs{\mut - \mu} > \frac{\sigma \log(1/\delta)}{\nb}} \leq \delta,$
  and it performs $\bo{\nb \log^{3/2}(\nb) \log\log(\nb)}$ quantum experiments.
\end{rtheorem}

Then we turn our attention to the $(\eps,\delta)$-approximation problem in Section~\ref{Sec:epsdelta}. In case we have an upper bound $\ch$ on the coefficient of variation $\abs{\sigma/\mu}$, we directly use our sub-Gaussian estimator to obtain an algorithm that performs $\wbo*{\frac{\ch}{\eps}\log(1/\delta)}$ quantum experiments (Corollary~\ref{Cor:Chebyshev}). Next, we consider the more subtle parameter-free setting where there is no prior information about the input random variable, except that it is distributed in $[0,1]$. In this case, the number of experiments is chosen \emph{adaptively}, and the bound we get is stated in expectation.

\begin{rtheorem}[Theorem~\ref{Thm:SeqEstim} {\normalfont (Restated)}]
  There exists a quantum algorithm with the following properties.
  Let~$X$ be a \qrv\ distributed in $[0,1]$ with mean $\mu$ and variance $\sigma^2$, and set as input two reals $\eps, \delta \in (0,1)$. Then, the algorithm outputs a mean estimate~$\mut$ such that $\pr*{\abs{\mut - \mu} > \eps \mu} \leq \delta$, and it performs
    $\wbo[\big]{\pt[\big]{\frac{\sigma}{\eps \mu} + \frac{1}{\sqrt{\eps \mu}}} \log(1/\delta)}$
  quantum experiments in expectation.
\end{rtheorem}

Finally, we prove several lower bounds in Section~\ref{Sec:LowerBoundMean} that match the complexity of the above estimators. We also consider the weaker input model where one is given copies of a quantum state encoding the distribution of $X$. We prove that no quantum speedup is achievable in this setting (Theorem~\ref{Thm:StateBLower}).


\subsection{Proof overview}
\label{Sec:MeanOverview}

\paragraph*{Sub-Gaussian estimator.} Our approach (Theorem~\ref{Thm:SubGaussian}) combines several ideas used in previous classical and quantum mean estimators. In this section, we simplify the exposition by assuming that the random variable $X$ is non-negative and by replacing the variance~$\sigma^2$ with the second moment $\ex{X^2}$. We also take the failure probability $\delta$ to be a small constant. Our starting point is a variant of the \emph{truncated mean estimators}~\cite{Bic65j,BCL13j,LM19j}. Truncation is a process that consists of replacing the samples larger than some threshold value with a smaller number. This has the effect of reducing the tail of the distribution, but also of changing its expectation. Here we study the effect of replacing the values larger than some threshold $b$ with $0$, which corresponds to the new random variable $Y = X \ind{X \leq b}$. We consider the following classical sub-Gaussian estimator that we were not able to find in the literature: set $b = \sqrt{\nb \ex{X^2}}$ and compute the empirical mean of~$\nb$ samples from $Y$. By a simple calculation, one can prove that the expectation of the removed part is at most $\ex{X-Y} \leq \ex{X^2}/b = \sqrt{\ex{X^2}/\nb}$. Moreover, using Bernstein's inequality and the boundedness of $Y$, the error between the output estimate and $\ex{Y}$ is on the order of $\sqrt{\ex{X^2}/\nb}$. These two facts together imply that the overall error for estimating~$\ex{X}$ is indeed of a sub-Gaussian type. This approach can be carried out in the quantum model by performing the truncation in superposition. This is similar to what is done in previous quantum mean estimators~\cite{Hei02j,Mon15j,HM19c}. In order to obtain a quantum speedup, one must balance the truncation level differently by taking $b = \nb \sqrt{\ex{X^2}}$. Then, by a clever use of amplitude estimation discovered by Heinrich~\cite{Hei02j} (see also~\cite[Proposition A.1]{HM18p}), the expectation of~$Y$ can be estimated with an error on the order of $\sqrt{\ex{Y^2}}/\nb + \max\pt{Y}/\nb^2 \leq 2\sqrt{\ex{X^2}}/\nb$. The main drawback of this estimator is that it requires the knowledge of~$\ex{X^2}$ to perform the truncation. In previous work~\cite{Hei02j,Mon15j,HM19c}, the authors made further assumptions on the variance to be able to approximate $b$. Here, we overcome this issue by choosing the truncation level $b$ differently. Borrowing ideas from classical estimators~\cite{LM19j}, we define $b$ as the quantile value that satisfies $\pr{X \geq b} = 1/\nb^2$. This quantile is always smaller than the previous threshold value $\nb \sqrt{\ex{X^2}}$. Moreover, it can be shown that the removed part $\ex{X-Y}$ is still on the order of $\sqrt{\ex{X^2}}/\nb$. We give a new quantum algorithm for approximating this quantile with roughly $\nb$ quantum experiments (Theorem~\ref{Thm:Quantile}), whereas it would require $\nb^2$ random experiments classically. Our quantile estimation algorithm builds upon the quantum minimum finding algorithm of D{\"{u}}rr and H{\o}yer~\cite{DH96p,vAGGdW20ja} and the $k$th-smallest element finding algorithm of Nayak and Wu~\cite{NW99c}. Importantly, it does not require any knowledge about $\ex{X^2}$.

\paragraph*{$(\eps,\delta)$-Approximation without side information.} We follow an approach similar to that of a classical estimator described in~\cite{DKLR00j}. Our algorithm (Theorem~\ref{Thm:SeqEstim}) uses the quantum sub-Gaussian estimator and the quantum \emph{sequential Bernoulli estimator} described in Proposition~\ref{Prop:SeqZeroOne}. The latter estimator can estimate the mean~$\mu$ of a random variable $X$ distributed in $[0,1]$ with constant relative error by performing $\bo{1/\sqrt{\mu}}$ quantum experiments in expectation. The first step of the $(\eps,\delta)$-approximation algorithm is to compute a rough estimate $\muh$ of $\mu$ with the sequential Bernoulli estimator. Then, the variance $\sigma^2$ of $X$ is estimated by using again the sequential Bernoulli estimator on the random variable $(X-X')/2$ (where $X'$ is an independent copy of $X$). The latter estimation is stopped if it uses more than $\bo{1/\sqrt{\eps \muh}}$ quantum experiments. We show that if $\sigma^2 \geq \om{\eps \mu}$ then the computation is not stopped and the resulting estimate $\td{\sigma}^2$ is close to $\sigma^2$ with high probability. Otherwise, it is stopped with high probability and we set $\td{\sigma} = 0$. Finally, the quantum sub-Gaussian estimator is used with the parameter $\nb \approx \max\pt[\big]{\frac{\td{\sigma}}{\eps \muh}, \frac{1}{\sqrt{\eps \muh}}}$ to obtain a refined estimate $\mut$ of $\mu$. The choice of the first (resp. second) term in the maximum value implies that $\abs{\mut - \mu} \leq \eps \mu$ with high probability when the variance $\sigma^2$ is larger (resp. smaller) than~$\eps \mu$. In order to upper bound the expected number of experiments performed by this estimator, we show in Proposition~\ref{Prop:SeqZeroOne} that the estimates~$\muh$ and~$\td{\sigma}$ obtained with the sequential Bernoulli estimator satisfy the expectation bounds $\ex{1/\muh} \leq 1/\mu$, $\ex{\td{\sigma}} \leq \sigma$ and $\ex{1/\sqrt{\muh}} \leq 1/\sqrt{\mu}$.

\paragraph*{Lower bounds.} We sketch the proof of optimality of the quantum sub-Gaussian estimator (Theorem~\ref{Thm:SubGLower}). The lower bound is proved in the stronger quantum query model, which allows us to extend it to all the other models mentioned in Section~\ref{Sec:ModelSampling}. Our approach is inspired by the truncation level chosen in the algorithm. Given $\sigma$ and $\nb$, we consider the two distributions $p_0$ and $p_1$ that output respectively $\frac{\nb\sigma}{\sqrt{1-1/\nb^2}}$ and $\frac{- \nb\sigma}{\sqrt{1-1/\nb^2}}$ with probability~$1/\nb^2$, and $0$ otherwise. The two distributions have variance~$\sigma^2$ and the distance between their means is larger than $\frac{2\sigma}{\nb}$. Thus, any estimator that satisfies the bound $\pr*{\abs{\mut - \mu} > \frac{\sigma}{\nb}} \leq \frac{1}{3}$ can distinguish between $p_0$ and $p_1$ with constant success probability. However, we show by a reduction to Quantum Search that it requires at least $\om{\nb}$ quantum experiments to distinguish between two distributions that differ with probability at most~$1/\nb^2$.

\section{Model of input}
\label{Sec:ModelSampling}
The input to the mean estimation problem is represented by a real-valued random variable~$X$ defined on some probability space. A classical estimator accesses this input by obtaining~$\nb$ \emph{i.i.d} samples of~$X$. In this section, we describe the access model for quantum estimators and we compare it to previous models suggested in the literature. We only consider finite probability spaces for finite encoding reasons. First, we recall the definition of a random variable, and we define a classical model of access called a \emph{random experiment}.

\begin{definition}[\sc Random variable]
  \label{Def:RandVar}
  A finite \emph{random variable} is a function $X : \Omega \ra E$ for some probability space $(\Omega,p)$, where $\Omega$ is a finite sample set, $p : \Omega \ra [0,1]$ is a probability mass function and $E \subset \R$ is the support of $X$. As is customary, we will often omit to mention $(\Omega,p)$ when referring to the random variable $X$.
\end{definition}

\begin{definition}[\sc Random experiment]
  \label{Def:rExp}
  Given a random variable $X$ on a probability space $(\Omega,p)$, we define a \emph{random experiment} as the process of drawing a sample $\omega \in \Omega$ according to $p$ and observing the value of $X(\omega)$.
\end{definition}

We now introduce the concept of ``\qrv'' to represent a quantum process that outputs a real number.

\begin{definition}[\sc \qrv]
  \label{Def:QVar}
  A \emph{q-variable} is a triple $(\Hil,U,M)$ where $\Hil$ is a finite-dimensional Hilbert space, $U$ is a unitary transformation on $\Hil$, and $M = \set{M_x}_{x \in E}$ is a projective measurement on $\Hil$ indexed by a finite set $E \subset \R$. Given a random variable~$X$ on a probability space $(\Omega,p)$, we say that a q-variable $(\Hil,U,M)$ \emph{generates} $X$ when,
  \begin{enumresult}
    \item $\Hil$ is a finite-dimensional Hilbert space with some basis $\set{\ket{\omega}}_{\omega \in \Omega}$ indexed by $\Omega$.
    \item $U$ is a unitary transformation on $\Hil$ such that $U \qub = \sum_{\omega \in \Omega} \sqrt{p(\omega)} \ket{\omega}$.
    \item $M = \set{M_x}_{x}$ is the projective measurement on $\Hil$ defined by $M_x = \sum_{\omega : X(\omega) = x} \proj{\omega}$.
  \end{enumresult}
  A random variable $X$ is a \emph{\qrv} if it is generated by some q-variable $(\Hil,U,M)$.
\end{definition}

We stress that the sample space $\Omega$ may not be known explicitly, and we do not assume that it is easy to perform a measurement in the $\set{\ket{\omega}}_{\omega \in \Omega}$ basis for instance. Often, we are given a unitary~$U$ such that $U \qub = \sum_{x \in E} \sqrt{p(x)} \ket{\psi_x} \ket{x}$ for some unknown garbage unit state~$\ket{\psi_x}$, together with the measurement $M = \set{I \otimes \proj{x}}_{x \in E}$. In this case, we can consider the \qrv~$X$ defined on the probability space $(\Omega,p)$ where $\Omega = \set{\ket{\psi_x} \ket{x}}_{x \in E}$ and $X(\ket{\psi_x} \ket{x}) = x$.

We further assume that there exist two quantum oracles, defined below, for obtaining information on the function $X : \Omega \ra E$. These two oracles can be efficiently implemented if we have access to a quantum evaluation oracle $\ket{\omega}\qub \mapsto \ket{\omega}\ket{X(\omega)}$ for instance. The rotation oracle (Assumption~\ref{Assp:Rot}) has been extensively used in previous quantum mean estimators~\cite{Ter99d,BDGT11p,Mon15j,HM19c}. The comparison oracle (Assumption~\ref{Assp:Comp}) is needed in our work to implement the quantile estimation algorithm.


\begin{asspt}{A}[\sc Comparison oracle]
  \label{Assp:Comp}
  Given a \qrv\ $X$ on a probability space $(\Omega,p)$, and any two values $a,b \in \R \cup \set{-\infty,+\infty}$ such that $a < b$, there is a unitary operator $C_{a,b}$ acting on $\Hil \otimes \C^2$ such that for all $\omega \in \Omega$,
  \begin{align*}
    C_{a,b} (\ket{\omega}\ket{0}) =
      \begin{cases}
        \ket{\omega}\ket{1} & \text{when $a < X(\omega) \leq b$,} \\
        \ket{\omega}\ket{0} & \text{otherwise.}
      \end{cases}
  \end{align*}
\end{asspt}

\begin{asspt}{B}[\sc Rotation oracle]
  \label{Assp:Rot}
  Given a \qrv\ $X$ on a probability space $(\Omega,p)$, and any two values $a,b \in \R \cup \set{-\infty,+\infty}$ such that $a < b$, there is a unitary operator $R_{a,b}$ acting on $\Hil \otimes \C^2$ such that for all $\omega \in \Omega$,
    \begin{align*}
      R_{a,b} (\ket{\omega}\ket{0}) =
        \begin{cases}
          \ket{\omega} \pt[\bigg]{\sqrt{1-\abs*{\frac{X(\omega)}{b}}}\ket{0} + \sqrt{\abs*{\frac{X(\omega)}{b}}} \ket{1}} & \text{when $a < X(\omega) \leq b$,} \\
          \ket{\omega}\ket{0} & \text{otherwise.}
        \end{cases}
    \end{align*}
\end{asspt}

We now define the measure of complexity used to count the number of accesses to a \qrv, which are referred to as \emph{quantum experiments}.

\begin{definition}[\sc Quantum experiment]
  \label{Def:qExp}
  Let $X$ be a \qrv\ that satisfies Assumptions~\ref{Assp:Comp} and~\ref{Assp:Rot}. Let $(\Hil,U,M)$ be a q-variable that generates $X$. We define a \emph{quantum experiment} as the process of applying any of the unitaries $U$, $C_{a,b}$, $R_{a,b}$ (for any values of~$a < b$), their inverses or their controlled versions, or performing a measurement according to~$M$.
\end{definition}

Note that a random experiment (Definition~\ref{Def:rExp}) can be simulated with two quantum experiments by computing the state $U \qub$ and measuring it according to $M$. We briefly mention two other possible input models. First, some authors~\cite{Gro98c,NW99c,Hei02j,BHH11j,CFMdW10c,BDGT11p,LW19j} consider the stronger query model where $p$ is the uniform distribution and a quantum evaluation oracle is provided for the function $\omega \mapsto X(\omega)$. A second model tackles the problem of \emph{learning from quantum states}~\cite{BJ99j,AdW18j,ABC20c}, where the input consists of several copies of $\sum_{x \in E} \sqrt{\pr{X = x}} \ket{x}$ (we do not have access to a unitary preparing that state). We show in Theorem~\ref{Thm:StateBLower} that no quantum speedup is achievable for our problem in the latter setting.

\section{Quantile estimation}
\label{Sec:Quantile}
In this section, we present a quantum algorithm for estimating the quantiles of a finite random variable $X$. This is a key ingredient for the sub-Gaussian estimator of Section~\ref{Sec:SubGaussian}. For the convenience of reading, we define a quantile in the following non-standard way (the cumulative distribution function is replaced with its complement).

\begin{definition}[\sc Quantile]
  Given a discrete random variable $X$ and a real $p \in [0,1]$, the \emph{quantile} of order $p$ is the number
    $Q(p) = \sup\set{x \in \R : \pr{X \geq x} \geq p}$.
\end{definition}

Our result is inspired by the minimum finding algorithm of D{\"{u}}rr and H{\o}yer~\cite{DH96p} and its generalization in~\cite{vAGGdW20ja}. The problem of estimating the quantiles of a set of numbers under the \emph{uniform} distribution was studied before by Nayak and Wu~\cite{NW99c,Nay99d}. We differ from that work by allowing arbitrary distributions, and by not using the amplitude estimation algorithm. On the other hand, we restrict ourselves to finding a constant factor estimate, whereas \cite{NW99c,Nay99d} can achieve any wanted accuracy.

The idea behind our algorithm is rather simple: if we compute a sequence of values $-\infty = y_0 \leq y_1 \leq y_2 \leq y_3 \leq \dots$ where each~$y_{j+1}$ is sampled from the distribution of~$X$ conditioned on $y_{j+1} \geq y_j$, then when $j \simeq \log(1/p)$ the value of $y_j$ should be close to the quantile $Q(p)$. The complexity of sampling each $y_j$ is on the order of $1/\pr{X \geq y_j}$ classically, but it can be done quadratically faster in the quantum setting. We analyze a slightly different algorithm, where the sequence of samples is strictly increasing and instead of stopping after roughly $\log(1/p)$ iterations we count the number of experiments performed by the algorithm and stop when it reaches a value close to $1/\sqrt{p}$. This requires showing that the times $T_j$ spent on sampling $y_j$ is neither too large nor too small with high probability, which is proved in the next lemma.

\begin{lemma}
  \label{Lem:IntervalSampl}
  There is a quantum algorithm such that, given a \qrv\ $X$ and a value $x \in \R \cup \set{-\infty,+\infty}$, it outputs a sample $y$ from the probability distribution of $X$ conditioned on $y > x$. If we let $T$ denote the number of quantum experiments performed by this algorithm, then there exist two universal constants $c_0 < c_1$ such that $\ex{T} \leq c_1/\sqrt{\pr{X > x}}$ and $\pr{T < c_0/\sqrt{\pr{X > x}}} \leq 1/10$.
\end{lemma}

\begin{proof}
  Let $(\Hil,U,M)$ be a q-variable generating $X$. We use the comparison oracle $C_{x,+\infty}$ from Assumption~\ref{Assp:Comp} to construct the unitary $V = C_{x,+\infty} (U \otimes I)$ acting on $\Hil \otimes \C^2$. By definition of $C_{x,+\infty}$ and $U$ (Section~\ref{Sec:ModelSampling}), we have that
  $
    V \qub
      = \sum_{\omega \in \Omega : X(\omega) \leq x} \sqrt{p(\omega)} \ket{\omega} \ket{0} + \sum_{\omega \in \Omega : X(\omega) > x} \sqrt{p(\omega)} \ket{\omega} \ket{1}
      = \sqrt{1 - \pr{X > x}} \ket{\phi_0} \ket{0} + \sqrt{\pr{X > x}} \ket{\phi_1} \ket{1}
  $
  for some unit states $\ket{\phi_0}, \ket{\phi_1}$ where $\ket{\phi_1} = \frac{1}{\sqrt{\pr{X > x}}} \sum_{\omega : X(\omega) > x} \sqrt{p(\omega)} \ket{\omega}$. The algorithm for sampling~$y$ conditioned on $y > x$ consists of two steps. First, we use the sequential amplitude amplification algorithm $\saamp(V,I \otimes \proj{1})$ from Theorem~\ref{Thm:SeqAA} on~$V$ to obtain the state $\ket{\phi_1}$. Next, we measure $\ket{\phi_1}$ according to $M$. The claimed properties follow directly from Theorem~\ref{Thm:SeqAA}.
\end{proof}

We use the next formula for the probability that a value $x$ occurs in the sequence $(y_j)_j$ defined before. This lemma is adapted from \cite[Lemma 1]{DH96p}.

\begin{lemma}[Lemma 47 in~\cite{vAGGdW20ja}]
  \label{Lem:infiniteSampling}
  Let $X$ be a discrete random variable. Consider the increasing sequence of random variables $Y_0,Y_1,Y_2,\dots$ where $Y_0$ is a fixed value and~$Y_{j+1}$ for $j \geq 0$ is a sample drawn from $X$ conditioned on $Y_{j+1} > Y_j$. Then, for any $x,y \in \R$,
     \begin{align*}
       \pr{x \in \set{Y_1,Y_2,\dots} \given Y_0 = y} =
         \begin{cases}
           \pr{X = x \given X \geq x} & \text{when $x > y$,} \\
           0 & \text{otherwise.}
         \end{cases}
     \end{align*}
\end{lemma}


The quantile estimation algorithm is described in Algorithm~\ref{Alg:Quantile} and its analysis is provided in the next theorem.

\algobox{Alg:Quantile}{Quantile estimation algorithm, $\quant(X,p,\delta)$.}{
\begin{enumerate}[leftmargin=*]
  \item Repeat the following steps for $i = 1, 2, \dots, \ceil{6\log(1/\delta)}$.
  \begin{enumerate}
    \item Set $y_0 = -\infty$ and initialize a counter $C = 0$ that is incremented each time a quantum experiment is performed.
    \item Set $j = 1$. Repeat the following process and interrupt it when $C = c'/\sqrt{p}$ (where $c'$ is a constant chosen in the proof of Theorem~\ref{Thm:Quantile}): sample an element $y_{j+1}$ from $X$ conditioned on $y_{j+1} > y_j$ by using the algorithm of Lemma~\ref{Lem:IntervalSampl}, set $j \leftarrow j+1$.
    \item Set $\td{Q}\super{i} = y_j$.
  \end{enumerate}
  \item Output $\td{Q} = \median\pt{\td{Q}\super{1},\dots,\td{Q}\super{\ceil{6\log(1/\delta)}}}$.
\end{enumerate}
}

\begin{theorem}[\sc Quantile estimation]
  \label{Thm:Quantile}
  Let $X$ be a \qrv. Given two reals $p, \delta \in (0,1)$, the approximate quantile $\td{Q}$ produced by the \emph{quantile estimation} algorithm $\quant(X,p,\delta)$ (Algorithm~\ref{Alg:Quantile}) satisfies
    \[Q(p) \leq \td{Q} \leq Q(cp)\]
  with probability at least $1-\delta$, where $c < 1$ is a universal constant. The algorithm performs $\bo[\Big]{\frac{\log(1/\delta)}{\sqrt{p}}}$ quantum experiments.
\end{theorem}

\begin{proof}
  Let $c_0$, $c_1$ be the universal constants mentioned in Lemma~\ref{Lem:IntervalSampl}, and set $c = c_0^2/(c_1^2\sqrt{191})$ and $c' = 190c_1$. Fix $i$ and consider the sequence $(y_j)_{j \geq 0}$ that would be computed during the $i$-th execution of steps 1.a-1.c if the stopping condition on $C$ was removed. We prove that immediately after the $c'/\sqrt{p}$-th quantum experiment is performed (which may occur during the computation of $y_{j+1}$), the current value of $y_j$ satisfies $Q(p) \leq y_j \leq Q(cp)$ with probability at least $(9/10)^2$. The analysis is done in two parts.

  First, let $x^- = Q(p)$ and denote by $T^-$ the number of experiments performed until $y_j$ becomes larger than or equal to $x^-$. According to Lemma~\ref{Lem:infiniteSampling}, the probability that a given $x$ occurs in the sequence $(y_j)_{j \geq 0}$ is equal to $\pr{X = x \given X \geq x}$. Moreover, using Lemma~\ref{Lem:IntervalSampl}, the expected number of experiments performed at step 1.b when $y_j = x$ is at most $c_1/\sqrt{\pr{X > x}}$. Consequently, we have
   \[\ex{T^-} \leq c_1 \sum_{x < x^-} \frac{\pr{X = x \given X \geq x}}{\sqrt{\pr{X > x}}}.\]
  Suppose that $Q(1) \neq x^-$ (otherwise $T^- = 0$). We upper bound the above sum by splitting it into several parts as follows. Define $Q_k = Q(2^{-k})$ for $k \geq 0$ and let $\ell$ be the largest integer such that $Q(2^{-\ell}) < x^-$. For each $1 \leq k \leq \ell$ such that $Q_{k-1} \neq Q_k$, we have
   \begin{align*}
    \sum_{Q_{k-1} \leq x < Q_k} \frac{\pr{X = x \given X \geq x}}{\sqrt{\pr{X > x}}}
      & \leq \frac{1}{\sqrt{\pr{X > Q_{k-1}}}} + \sum_{Q_{k-1} < x < Q_k} \frac{\pr{X = x}}{\pr{X > x}^{3/2}} \\
      & \leq \frac{1}{\sqrt{\pr{X \geq Q_k}}} + \frac{\pr{X > Q_{k-1}}}{\pr{X \geq Q_k}^{3/2}} \\
      & \leq \frac{1}{\sqrt{2^{-k}}} + \frac{2^{-(k-1)}}{2^{-3k/2}} \\
      & \leq 2^{k/2 + 2}.
   \end{align*}
  Similarly, $\sum_{Q_{\ell} \leq x < x^-} \frac{\pr{X = x \given X \geq x}}{\sqrt{\pr{X > x}}} \leq 2^{\ell/2} + 2^{- \ell + 1}/p^{3/2}$. Thus,
    $\ex{T^-} \leq c_1 \pt[\big]{\sum_{k=1}^{\ell} 2^{k/2 + 2} + 2^{\ell/2} + 2^{- \ell + 1}/p^{3/2}} \leq 19c_1/\sqrt{p}$
  where we used that $\log(1/p) - 1 \leq \ell < \log(1/p)$ since $Q_{\ell} < Q(p) \leq Q_{\ell+1}$. By Markov's inequality, $\pr{T^- \leq 190c_1/\sqrt{p}} \geq 9/10$.

  Secondly, let $x^+ = Q(cp)$ and denote by $T^+$ the number of experiments performed at step 1.b to sample $y_{j+1}$ when $y_j \geq x^+$. According to Lemma~\ref{Lem:IntervalSampl}, we have $\pr{T^+ \geq c_0/\sqrt{\pr{X > y_j}}} \geq 9/10$. Moreover, $\pr{X > y_j} \leq cp = c_0^2/(c_1^2\sqrt{191})p$ by definition of $x^+$. Thus, $\pr{T^+ \geq 191c_1/\sqrt{p}} \geq 9/10$.

  We conclude that step 1.b is interrupted when the value $\td{Q}\super{i}$ satisfies $Q(p) \leq \td{Q}\super{i} \leq Q(cp)$ with probability at least $(9/10)^2$. Thus, by the Chernoff bound, the output $\td{Q}$ satisfies $Q(p) \leq \td{Q} \leq Q(cp)$ with probability at least $1-\delta$. The total number of experiments is guaranteed to be $\bo{\log(1/\delta)/\sqrt{p}}$ by our use of the counter $C$.
\end{proof}

\section{Sub-Gaussian estimator}
\label{Sec:SubGaussian}

In this section, we present the main quantum algorithm for estimating the mean of a random variable with a near-quadratic speedup over the classical sub-Gaussian estimators. Our result uses the following \emph{Bernoulli estimator}, which is a well-known adaptation of the amplitude estimation algorithm to the mean estimation problem~\cite{BHMT02j,Ter99d,Mon15j}. The Bernoulli estimator allows us to estimate the mean of the truncated random variable $X \ind{a < X \leq b}$ for any $a,b$.

\begin{proposition}[\sc Bernoulli estimator]
  \label{Prop:ZeroOne}
  There exists a quantum algorithm, called the \emph{Bernoulli estimator}, with the following properties.
  Let $X$ be a \qrv\ and set as input a time parameter $\nb \geq 0$, two range values $0 \leq a < b$, and a real $\delta \in (0,1)$ such that $\nb \geq \log(1/\delta)$. Then, the Bernoulli estimator $\bern(X,\nb,a,b,\delta)$ outputs a mean estimate $\mut_{a,b}$ of $\mu_{a,b} = \ex{X \ind{a < X \leq b}}$ such that $\abs{\mut_{a,b} - \mu_{a,b}} \leq \frac{\sqrt{b \mu_{a,b}} \log(1/\delta)}{\nb} + \frac{b\log(1/\delta)^2}{\nb^2}$. It performs~$\bo{\nb}$ quantum experiments.
\end{proposition}

\begin{proof}
  Let $(\Hil,U,M)$ be a q-variable generating $X$. Using the rotation oracle $R_{a,b}$ from Assumption~\ref{Assp:Rot}, we define the unitary algorithm $V = R_{a,b} (U \otimes I)$ acting on $\Hil \otimes \C^2$. In order to simplify notations, let us first assume that the random variable $X$ is distributed in the interval $(a,b)$. Then, $\mu = \mu_{a,b}$ and by definition of $R_{a,b}$ and $U$ (Section~\ref{Sec:ModelSampling}) we have,
  \begin{align*}
    V \qub
     & = \sum_{\omega \in \Omega} \sqrt{p(\omega)} \ket{\omega} \pt*{\sqrt{1-\frac{X(\omega)}{b}}\ket{0} + \sqrt{\frac{X(\omega)}{b}} \ket{1}} \\
     & = \sqrt{1 - \frac{\mu}{b}} \pt[\Bigg]{\sum_{\omega \in \Omega} \sqrt{\frac{p(\omega) (b - X(\omega))}{b - \mu}}  \ket{\omega}} \ket{0} + \sqrt{\frac{\mu}{b}} \pt*{\sum_{\omega \in \Omega} \sqrt{\frac{p(\omega) X(\omega)}{\mu}} \ket{\omega}} \ket{1}.
  \end{align*}
  Thus, there exist some unit states $\ket{\psi_0}, \ket{\psi_1}$ such that $V \qub = \sqrt{1 - \frac{\mu}{b}}\ket{\psi_0} + \sqrt{\frac{\mu}{b}} \ket{\psi_1}$ and $(I \otimes \proj{1}) V \qub = \sqrt{\frac{\mu}{b}} \ket{\psi_1}$. If $X$ takes values outside the interval $(a,b)$ then the same result holds with $\mu_{a,b}$ in place of $\mu$ and a different definition of $\ket{\psi_0}, \ket{\psi_1}$.

  Consider the output $\td{v}$ of the amplitude estimation algorithm $\aest\pt[\big]{V, \Pi, \ceil[\big]{\frac{2\pi \nb}{\log(1/\delta)}}}$ (Theorem~\ref{Thm:AE}) where $\Pi = I \otimes \proj{1}$. Then, the estimate $b \td{v}$ satisfies the statement of the proposition with probability $8/\pi^2$ by Theorem~\ref{Thm:AE}. The Bernoulli estimator consists of running $\ceil{6\log(1/\delta)}$  copies of $\aest\pt[\big]{V,\Pi,\ceil[\big]{\frac{2\pi \nb}{\log(1/\delta)}}}$ and outputting the median of the results. The success probability is at least $1-\delta$ by the Chernoff bound.
\end{proof}

The Bernoulli estimator can estimate the mean of a non-negative \qrv~$X$ by setting $a = 0$ and $b = \max X$. However, its performance is worse than that of the classical sub-Gaussian estimators when the maximum of $X$ is large compared to its variance. Our quantum sub-Gaussian estimator (Algorithm~\ref{Alg:SubGaussian}) uses the Bernoulli estimator in a more subtle way, and in combination with the quantile estimation algorithm.

\algobox{Alg:SubGaussian}{Sub-Gaussian estimator, $\subgauss(X,\nb,\delta)$.}{
\begin{enumerate}[leftmargin=*]
  \item Set $k = \log \nb$ and $\mb = d \nb \sqrt{\log \nb} \frac{\log(9k/\delta)}{\log(1/\delta)}$, where $d > 1$ is a constant chosen in the proof of Theorem~\ref{Thm:SubGaussian} (if $k$ is not an integer, round $\nb$ to the next power of two).
  \item Compute the median $\eta$ of $\lceil 30 \log(2/\delta) \rceil$ classical samples from $X$ and define the non-negative random variables
  \[Y^+ = (X - \eta) \ind{X \geq \eta} \quad \text{and} \quad Y^- = - (X - \eta) \ind{X \leq \eta}.\]
  \item Compute an estimate $\mut_{Y_+}$ of $\ex{Y_+}$ and an estimate $\mut_{Y_-}$ of $\ex{Y_-}$ by executing the following steps with $Y := Y_+$ and $Y := Y_-$ respectively:
  \begin{enumerate}
    \item Compute an estimate $\td{Q}$ of the quantile of order $p = \pt*{\frac{\log(1/\delta)}{6\nb}}^2$ of $Y$ with failure probability $\delta/8$ by using the \hyperref[Thm:Quantile]{quantile estimation} algorithm $\quant(Y,p,\delta/8)$.
    \item Define $a_{-1} = 0$ and $a_{\ell} = \frac{2^{\ell}}{\nb} \td{Q}$ for $\ell \geq 0$. Compute an estimate $\mut_{\ell}$ of $\ex{Y \ind{a_{\ell - 1} < Y \leq a_{\ell}}}$ with failure probability $\delta/(9k)$ for each $0 \leq \ell \leq k$, by using the \hyperref[Prop:ZeroOne]{Bernoulli estimator} $\bern(Y,\mb,a_{\ell - 1},a_{\ell},\delta/(9k))$ with $\mb$ quantum experiments.
    \item Set $\mut_Y = \sum_{\ell = 0}^{k} \mut_{\ell}$.
  \end{enumerate}
  \item Output $\mut = \eta + \mut_{Y_+} - \mut_{Y_-}$.
\end{enumerate}
}

\begin{theorem}[\sc Sub-Gaussian estimator]
  \label{Thm:SubGaussian}
  Let $X$ be a \qrv\ with mean~$\mu$ and variance $\sigma^2$. Given a time parameter $\nb$ and a real $\delta \in (0,1)$ such that $\nb \geq \log(1/\delta)$, the \emph{sub-Gaussian estimator} $\subgauss(X,\nb,\delta)$ (Algorithm~\ref{Alg:SubGaussian}) outputs a mean estimate~$\mut$ such that,
    \[\pr*{\abs{\mut - \mu} \leq \frac{\sigma \log(1/\delta)}{\nb}} \geq 1-\delta.\]
  The algorithm performs $\bo{\nb \log^{3/2}(\nb) \log\log(\nb)}$ quantum experiments.
\end{theorem}

\begin{proof}
  First, by standard concentration inequalities, the median $\eta$ computed at step 2 satisfies $\abs{\eta - \mu} \leq 2\sigma$ with probability at least $1-\delta/2$. Moreover, if $\abs{\eta - \mu} \leq 2\sigma$ then $\sqrt{\ex{(X-\eta)^2}} = \sqrt{\ex{(X - \mu + \mu - \eta)^2}} \leq \sqrt{\ex{(X - \mu)^2}} + \abs{\mu - \eta} \leq 3\sigma$, by using the triangle inequality. Below we prove that for any non-negative random variable $Y$ the estimate $\mut_Y$ of $\mu_Y = \ex{Y}$ computed at step 3 satisfies
    \begin{equation}
      \label{Eq:subGaussian}
      \abs{\mut_Y - \mu_Y} \leq \frac{\sqrt{\ex{Y^2}} \log(1/\delta)}{5\nb}
    \end{equation}
  with probability at least $1-\delta/4$. Using the fact that $X = \eta + Y_+ - Y_-$ and $(X-\eta)^2 = Y_+^2 + Y_-^2$, we can conclude that
   \[\abs{\mut - \mu}
      \leq \frac{\pt*{\sqrt{\ex{Y_+^2}}+\sqrt{\ex{Y_-^2}}} \log(1/\delta)}{5\nb}
      \leq \frac{\sqrt{2\ex{(X-\eta)^2}} \log(1/\delta)}{5\nb}
      \leq \frac{\sigma \log(1/\delta)}{\nb}\]
   with probability at least $1 - \delta$. The algorithm performs $\bo{\log(1/\delta)} \leq \bo{\nb}$ classical experiments during step 2, $\bo{\log(1/\delta)/\sqrt{p}} \leq \bo{\nb}$ quantum experiments during step 3.a, and $\bo{k \mb} \leq \bo{\nb \log^{3/2}(\nb) \log\log(\nb)}$ quantum experiments during step 3.b.

  We now turn to the proof of Equation~(\ref{Eq:subGaussian}). We make the assumption that all the subroutines used in step 3 are successful, which is the case with probability at least $(1 - \delta/8)(1-\delta/(9k))^{k+1} \geq 1 - \delta/4$. First, according to Theorem~\ref{Thm:Quantile}, we have $Q(p) \leq \td{Q} \leq Q(cp)$ for some universal constant $c$. It implies that $cp \leq \pr{Y \geq Q(cp)} \leq \pr{Y \geq \td{Q}} \leq \ex{Y^2}/\td{Q}^2$, where the first two inequalities are by definition of the quantile function~$Q$, and the last inequality is a standard fact. Consequently, by our choice of $p$,
    \begin{equation}
      \label{Eq:quantileGaussian}
      \td{Q} \leq \frac{6\nb \sqrt{\ex{Y^2}}}{\sqrt{c}\log(1/\delta)}.
    \end{equation}
  Next, we upper bound the expectation of the part of $Y$ that is above the largest threshold $a_k = \td{Q}$ considered in step 3.b. By Cauchy--Schwarz' inequality, we have $\ex{Y \ind{Y > \td{Q}}} \leq \sqrt{\ex{Y^2} \pr{Y > \td{Q}}}$. Moreover, by definition of $Q$, $\pr{Y > \td{Q}} \leq \pr{Y > Q(p)} \leq p$. Thus,
    \begin{equation}
      \label{Eq:tailGaussian}
      \ex{Y \ind{Y > \td{Q}}} \leq \frac{\sqrt{\ex{Y^2}}\log(1/\delta)}{6\nb}.
    \end{equation}
  The expectation of $Y$ is decomposed into the sum
    $\mu_Y = \sum_{\ell = 0}^k \mu_{\ell} + \ex{Y \ind{Y > a_k}}$,
  where $\mu_{\ell} = \ex{Y \ind{a_{\ell - 1} < Y \leq a_{\ell}}}$ is estimated at step 3.b. We have $\abs{\td{\mu}_{\ell} - \mu_{\ell}} \leq \frac{ \sqrt{a_{\ell} \mu_{\ell}} \log(1/\delta)}{d \nb \sqrt{\log \nb}} + \frac{a_{\ell} \log(1/\delta)^2}{d^2 \nb^2 \log \nb}$ for all $0 \leq \ell \leq k$ according to Proposition~\ref{Prop:ZeroOne}. Thus, by the triangle inequality,
  \begin{align*}
    \abs{\mut_Y - \mu_Y}
      & \leq \sum_{\ell = 0}^k \abs*{\mut_{\ell} - \mu_{\ell}} + \ex{Y \ind{Y > a_k}} \\
      & \leq \sum_{\ell = 0}^k \frac{\sqrt{a_{\ell} \mu_{\ell}} \log\pt[\big]{\frac{1}{\delta}}}{d \nb \sqrt{\log \nb}} + \sum_{\ell = 0}^k \frac{a_{\ell} \log\pt[\big]{\frac{1}{\delta}}^2}{d^2 \nb^2 \log \nb} + \ex{Y \ind{Y > a_k}} \\
      & \leq \frac{\td{Q} \log\pt[\big]{\frac{1}{\delta}}}{d \nb^2 \sqrt{\log \nb}} + \sum_{\ell = 1}^k \frac{\sqrt{2 \ex{Y^2 \ind{a_{\ell - 1} < Y \leq a_{\ell}}}} \log\pt[\big]{\frac{1}{\delta}}}{d \nb \sqrt{\log \nb}} + \frac{2 \td{Q} \log\pt[\big]{\frac{1}{\delta}}^2}{d^2 \nb^2 \log \nb} + \ex{Y \ind{Y > a_k}} \\
      & \leq  \frac{\sqrt{2 k} \sqrt{\sum_{\ell = 1}^k \ex{Y^2 \ind{a_{\ell - 1} < Y \leq a_{\ell}}}} \log\pt[\big]{\frac{1}{\delta}}}{d \nb \sqrt{\log \nb}} + \frac{3 \td{Q} \log\pt[\big]{\frac{1}{\delta}}^2}{d \nb^2 \sqrt{\log \nb}} + \ex{Y \ind{Y > a_k}} \\
      & \leq  \frac{\sqrt{2 k} \sqrt{\ex{Y^2}} \log\pt[\big]{\frac{1}{\delta}}}{d \nb \sqrt{\log \nb}} + \frac{3 \td{Q} \log\pt[\big]{\frac{1}{\delta}}^2}{d \nb^2 \sqrt{\log \nb}} + \ex{Y \ind{Y > a_k}} \\
      & \leq  \frac{\sqrt{2} \sqrt{\ex{Y^2}} \log\pt[\big]{\frac{1}{\delta}}}{d \nb} + \frac{18 \sqrt{\ex{Y^2}} \log\pt[\big]{\frac{1}{\delta}}}{\sqrt{c} d \nb \sqrt{\log \nb}} + \frac{\sqrt{\ex{Y^2}}\log\pt[\big]{\frac{1}{\delta}}}{6\nb} \\
      & \leq \frac{\sqrt{\ex{Y^2}} \log\pt[\big]{\frac{1}{\delta}}}{5\nb}
  \end{align*}
  where the third step uses $a_0 \mu_0 \leq a_0^2 = (\td{Q}/\nb)^2$ and $a_{\ell}\mu_{\ell} \leq (a_{\ell}/a_{\ell-1}) \ex{Y^2 \ind{a_{\ell - 1} < Y \leq a_{\ell}}} \leq 2 \ex{Y^2 \ind{a_{\ell - 1} < Y \leq a_{\ell}}}$ when $\ell \geq 1$, the fourth step uses the Cauchy--Schwarz inequality, the sixth step uses Equations~(\ref{Eq:quantileGaussian}) and~(\ref{Eq:tailGaussian}), and in the last step we choose $d = 600/\sqrt{c}$.
\end{proof}

\section{\texorpdfstring{$(\eps,\delta)$-Estimators}{(epsilon,delta)-Estimators}}
\label{Sec:epsdelta}
We study the $(\eps,\delta)$-approximation problem under two different scenarios. First, we consider the case where we know an upper bound $\ch$ on the coefficient of variation $\abs{\sigma/\mu}$. As a direct consequence of Theorem~\ref{Thm:SubGaussian} we obtain the following estimator that subsumes a similar result shown in~\cite{HM19c} for non-negative random variables.

\begin{corollary}[Relative estimator]
  \label{Cor:Chebyshev}
  There exists a quantum algorithm with the following properties.
  Let $X$ be a \qrv\ with mean $\mu$ and variance $\sigma^2$, and set as input a value $\ch \geq \abs{\sigma/\mu}$ and two reals $\eps,\delta \in (0,1)$. Then, the algorithm outputs a mean estimate $\mut$ such that $\pr*{\abs{\mut - \mu} > \eps \abs{\mu}} \leq \delta$ and it performs
    $\wbo*{\frac{\ch}{\eps}\log(1/\delta)}$
  quantum experiments.
\end{corollary}

\begin{proof}
  The algorithm runs the \hyperref[Thm:SubGaussian]{sub-Gaussian estimator} $\subgauss\pt[\big]{X,\frac{\ch}{\eps}\log(1/\delta),\delta}$.
\end{proof}

Next, we construct a parameter-free estimator that performs $\wbo[\big]{\pt[\big]{\frac{\sigma}{\eps \mu} + \frac{1}{\sqrt{\eps \mu}}} \log(1/\delta)}$ quantum experiments in expectation for any random variable distributed in $[0,1]$. We follow an approach similar to the classical $\mathcal{AA}$ algorithm described in~\cite{DKLR00j}. We first give a sequential estimator that approximates the mean with constant relative error and that performs $\bo{1/\sqrt{\mu}}$ quantum experiments in expectation. We use the term ``sequential'' in reference to sequential analysis techniques. The classical counterpart of this estimator is the Stopping Rule Algorithm in~\cite{DKLR00j}.

\begin{proposition}[\sc Sequential Bernoulli estimator]
  \label{Prop:SeqZeroOne}
  There is an algorithm, called the \emph{sequential Bernoulli estimator}, with the following properties.
  Let $X$ be a \qrv\ distributed in $[0,1]$ with mean $\mu$.
  Then, the sequential Bernoulli estimator $\sbern(X)$ outputs an estimate $\mut$ and performs a number $T$ of quantum experiments such that,
    \begin{enumerate}
      \item There is a universal constant $c \in (0,1)$ such that $\pr{\abs{\mut - \mu} \leq c \mu} \geq 7/8$.
      \item There is a universal constant $c'$ such that $\ex{T^2} = \ex{1/\mut} \leq c'/\mu$.
      \item There is a universal constant $c''$ such that $\ex{\sqrt{\mut}} \leq c''\sqrt{\mu}$.
    \end{enumerate}
\end{proposition}

\begin{proof}
  The algorithm is identical to the one of Proposition~\ref{Prop:ZeroOne} with $a = 0$ and $b = 1$, except that the amplitude estimation algorithm is replaced with the sequential amplitude estimation algorithm (Theorem~\ref{Thm:SeqAE}). The algorithm inherits the properties proved in Theorem~\ref{Thm:SeqAE}.
\end{proof}

The expected number of experiments performed by the sequential Bernoulli estimator is $\ex{T} \leq \sqrt{\ex{T^2}} \leq 1/\sqrt{\mu}$. The output $\mut$ of the sequential Bernoulli estimator can be used in the Bernoulli estimator (Proposition~\ref{Prop:ZeroOne}) with parameter $\nb = 8\log(1/\delta)/(\eps \sqrt{\mut})$ to solve the $(\eps,\delta)$-approximation problem. However, the expected number of experiments performed with this approach is $\bo{\log(1/\delta)/(\eps \sqrt{\mu})}$. We propose a better algorithm with an improved dependence on $\eps$. The algorithms uses the sequential Bernoulli estimator and the sub-Gaussian estimator.

\algobox{Alg:SeqEstim}{Sequential $(\eps,\delta)$-estimator.}{
\begin{enumerate}
  \item For $i = 1,\dots,32\log(1/\delta)$:
  \begin{enumerate}
    \item Compute an estimate $\mut_X$ of $\mu = \ex{X}$ by using the sequential Bernoulli estimator $\sbern(X)$ (Proposition~\ref{Prop:SeqZeroOne}). 
    \item Let $Y$ denote the random variable $(X-X')^2/2$ where $X'$ is independent from $X$ and identically distributed. Compute an estimate $\mut_Y$ of $\mu_Y = \ex{Y}$ by using the sequential Bernoulli estimator $\sbern(Y)$ (Proposition~\ref{Prop:SeqZeroOne}). 
    Stop the computation if it performs more than $\frac{c_1}{\sqrt{\eps \mut_X}}$ quantum experiments (where $c_1$ is a constant chosen in the proof of Theorem~\ref{Thm:SeqEstim}) and set $\mut_Y = 0$.
    \item Compute a second estimate $\mut_X\super{i}$ of $\mu$ by using the sub-Gaussian estimator $\subgauss(X,\nb,15/16)$ (Theorem~\ref{Thm:SubGaussian}) with $\nb = c_2 \max\pt*{\frac{\sqrt{\mut_Y}}{\eps \mut_X}, \frac{1}{\sqrt{\eps \mut_X}}}$ (where $c_2$ is a constant chosen in the proof of Theorem~\ref{Thm:SeqEstim}).
  \end{enumerate}
  \item Output $\mut = \median\pt*{\mut_X\super{1},\dots,\mut_X\super{32\log(1/\delta)}}$.
\end{enumerate}
}

\begin{theorem}[\sc Sequential relative estimator]
  \label{Thm:SeqEstim}
  Let $X$ be a \qrv\ distributed in $[0,1]$ with mean $\mu$ and variance $\sigma^2$. Given two reals $\eps, \delta \in (0,1)$ the estimate~$\mut$ output by the \emph{sequential relative estimator} (Algorithm~\ref{Alg:SeqEstim}) satisfies
    $\pr*{\abs{\mut - \mu} > \eps \mu} \leq \delta$.
  The algorithm performs
    $\wbo[\big]{\pt[\big]{\frac{\sigma}{\eps \mu} + \frac{1}{\sqrt{\eps \mu}}} \log(1/\delta)}$
  quantum experiments in expectation.
\end{theorem}

\begin{proof}
  We prove that, for a fixed value of $i$, the estimate $\mut_X\super{i}$ computed at step 1.c satisfies $\pr[\big]{\abs{\mut_X\super{i} - \mu} \leq \eps \mu} \geq 5/8$ and the number of experiments performed during its computation is $\wbo[\big]{\pt[\big]{\frac{\sigma}{\eps \mu} + \frac{1}{\sqrt{\eps \mu}}}}$ in expectation. The theorem follows by the Chernoff bound and the linearity of expectation.

  Let $c,c',c''$ denote the constants mentioned in Proposition~\ref{Prop:SeqZeroOne}, and set $c_1 = 16c' \sqrt{(1 + c)}$ and $c_2 = 4(1+c)/\sqrt{1-c}$. We assume that $\abs{\mut_X - \mu} \leq c \mu$ at step 1.a, which is the case with probability at least $7/8$ by Proposition~\ref{Prop:SeqZeroOne}. The analysis of steps 1.b and 1.c is split into two cases to show that $\pr[\big]{\abs{\mut_X\super{i} - \mu} \leq \eps \mu} \geq 5/8$. First, if $\sigma \leq \sqrt{\eps\mu}$, then we can ignore step 1.b and consider the second term in the $\max$ at step 1.c. By Theorem~\ref{Thm:SubGaussian}, the estimate $\mut_X\super{i}$ satisfies
    $\abs{\mut_X\super{i} - \mu}
        \leq \frac{4 \sigma}{c_2/\sqrt{\eps \mut_X}}
        \leq \frac{4\sqrt{1+c}}{c_2} \eps \mu
        \leq \eps \mu$
  with probability $15/16$. Secondly, if $\sigma \geq \sqrt{\eps\mu}$, then by Proposition~\ref{Prop:SeqZeroOne} and the fact that $\mu_Y = \sigma^2$, the estimate $\mut_Y$ computed at step 1.b satisfies $\abs{\mut_Y - \sigma^2} \leq c \sigma^2$ with probability $7/8$ if we remove the stopping condition. Since we assumed that $\mut_X \leq (1+c) \mu$, the computation is interrupted if it performs more than $\frac{c_1}{\sqrt{\eps \mut_X}} \geq \frac{c_1}{\sqrt{(1 + c)\mu_Y}} = \frac{16c'}{\sqrt{\mu_Y}}$ experiments. However, by Proposition~\ref{Prop:SeqZeroOne} and Markov's inequality, the number of experiments performed by the sequential Bernoulli estimator at step 1.b is at most $16c'/\sqrt{\mu_Y}$ with probability at least $15/16$. Consequently, we can assume that $\mut_Y \geq (1-c)\sigma^2$ with success probability at least $7/8 \cdot 15/16$. In this case, by considering the first term in the $\max$ at step 1.c, the estimate $\mut_X\super{i}$ satisfies
    $\abs{\mut_X\super{i} - \mu}
        \leq \frac{4 \sigma}{c_2\sqrt{\mut_Y}/(\eps \mut_X)}
        \leq \frac{4(1+c)}{c_2\sqrt{1-c}} \eps \mu
        \leq \eps \mu$
  with probability $15/16$. The overall success probability is at least $(7/8)^2 (15/16)^2 \geq 5/8$.

  We now analyse the expected number of quantum experiments performed during the computation of $\mut_X\super{i}$. Step 1.a performs $\bo{1/\sqrt{\mu}}$ experiments in expectation by Proposition~\ref{Prop:SeqZeroOne}. Step 1.b is stopped after $\bo{1/(\sqrt{\eps\mu})}$ experiments in expectation since $\ex{1/\sqrt{\mut_X}} \leq \bo{1/\sqrt{\mu}}$ by Proposition~\ref{Prop:SeqZeroOne}. Step 1.c performs $\wbo[\Big]{\max\pt[\Big]{\frac{\sqrt{\mut_Y}}{\eps \mut_X}, \frac{1}{\sqrt{\eps \mut_X}}}}$ experiments by Theorem~\ref{Thm:SubGaussian}. The estimates $\mut_Y$ and $\mut_X$ are independent if we ignore the stopping condition at step 1.b, in which case $\ex*{\frac{\sqrt{\mut_Y}}{\mut_X}} = \ex*{\frac{1}{\mut_X}}\ex{\sqrt{\mut_Y}} \leq \bo*{\frac{\sigma}{\mu}}$ by Proposition~\ref{Prop:SeqZeroOne}. The stopping condition can only decrease this quantity. Thus, step 1.c performs $\wbo[\big]{\max\pt[\big]{\frac{\sigma}{\eps \mu}, \frac{1}{\sqrt{\eps \mu}}}}$ experiments in expectation
\end{proof}

\section{Lower bounds}
\label{Sec:LowerBoundMean}
We prove several lower bounds for the mean estimation problem under different scenarios. In Section~\ref{Sec:LowerGauss}, we study the number of experiments that must be performed to estimate the mean with a sub-Gaussian error rate. In Section~\ref{Sec:LowerEpsDelta}, we study the number of experiments needed to solve the $(\eps,\delta)$-approximation problem. Finally, in Section~\ref{Sec:stateLower}, we consider the mean estimation problem in the state-based model, where the input consists of several copies of a quantum state encoding a distribution.


\subsection{Sub-Gaussian estimation}
\label{Sec:LowerGauss}

We show that the quantum sub-Gaussian estimator described in Theorem~\ref{Thm:SubGaussian} is optimal up to a polylogarithmic factor. We make use of the following lower bound for Quantum Search in the small-error regime.

\begin{proposition}[Theorem 4 in \cite{BCdWZ99c}]
  \label{Prop:QSearchLower}
  Let $N > 0$, $1 \leq K \leq 0.9N$ and $\delta \geq 2^{-N}$. Let $T(N,K,\delta)$ be the minimum number of quantum queries any algorithm must use to decide with failure probability at most~$\delta$ whether a function $f : [N] \ra \rn$ has $0$ or $K$ preimages of~$1$. Then, $T(N,K,\delta) \geq \om{\sqrt{N/K}\log(1/\delta)}$.
\end{proposition}

We construct two particular probability distributions that allow us to reduce the Quantum Search problem to the sub-Gaussian mean estimation problem.

\begin{theorem}
  \label{Thm:SubGLower}
  Let $\nb > 1$ and $\delta \in (0,1)$ such that $\nb \geq 2\log(1/\delta)$. Fix $\sigma > 0$ and consider the family $\mathcal{P}_{\sigma}$ of all \qrvs\ with variance $\sigma^2$. Let $T(\nb,\sigma,\delta)$ be the minimum number of quantum experiments any algorithm must perform to compute with failure probability at most $\delta$ a mean estimate $\mut$ such that $\abs{\mut - \mu} \leq \frac{\sigma \log(1/\delta)}{\nb}$ for any $X \in \mathcal{P}_{\sigma}$ with mean $\mu$. Then, $T(\nb,\sigma,\delta) \geq \om{\nb}$.
\end{theorem}

\begin{proof}
  Let $\mb = \frac{\nb}{\log(1/\delta)}$ and $b = \frac{\mb}{\sqrt{1-1/\mb^2}}\sigma$. We define the probability distribution $p_0$ with support~$\set{0,b}$ that takes value $b$ with probability $\frac{1}{\mb^2}$. Similarly, we define the probability distribution $p_1$ with support $\set{0,-b}$ that takes value $- b$ with probability $\frac{1}{\mb^2}$. The variance of each distribution is equal to $\sigma^2$. Moreover, the means $\mu_0$ and $\mu_1$ of the two distributions satisfy that,
    \begin{equation}
      \label{Eq:meanFar}
      \mu_0 - \mu_1 > 2 \frac{\sigma \log(1/\delta)}{\nb}.
    \end{equation}
  Let $N,K$ be two integers such that $N \geq \log(1/\delta)$ and $K/N = 1/\mb^2$ (assuming $\mb$ is rational). Let $F_0$ be the family of all functions $f : [N] \ra \rn$ with exactly~$K$ preimages of $1$. Similarly, let $F_1$ be the family of all functions $f : [N] \ra \set{-1,0}$ with exactly~$K$ preimages of $-1$. By using Proposition~\ref{Prop:QSearchLower}, it is easy to see that any algorithm that can distinguish between $f \in F_0$ and $f\in F_1$ with success probability $1-\delta$ must use at least $\om{\sqrt{N/K}\log(1/\delta)} = \om{\mb \log(1/\delta)} = \om{\nb}$ quantum queries to $f$. We associate with each function $f \in F_0 \cup F_1$ the q-variable $(\Hil,U,M)_{f}$ where $\Hil = \C^{N+2}$, $U \qub = \frac{1}{\sqrt{N}} \sum_{x \in [N]} \ket{x}\ket{f(x)}$, and $M = \set{I \otimes \proj{0}, I \otimes \proj{-1}, I \otimes \proj{1}}$. The random variable~$X$ generated by $(\Hil,U,M)_{f}$ is distributed according to $p_0$ if $f \in F_0$, and according to $p_1$ if $f \in F_1$. Moreover, one quantum experiment with respect to $X$ can be simulated with one quantum query to $f$. Consequently, any algorithm that can distinguish between a random variable distributed according to $p_0$ or $p_1$ with success probability $1-\delta$ must perform at least $\om{\nb}$ quantum experiments. On the other hand, by Equation~(\ref{Eq:meanFar}), if an algorithm can estimate the mean with an error rate smaller than $\frac{\sigma \log(1/\delta)}{\nb}$ then it can distinguish between $f \in F_0$ and $f \in F_1$. Thus, $T(\nb,\sigma,\delta) \geq \om{\nb}$.
\end{proof}


\subsection{\texorpdfstring{$(\eps,\delta)$-Estimation}{(epsilon,delta)-Estimation}}
\label{Sec:LowerEpsDelta}

We consider the $(\eps,\delta)$-estimation problem in the parameter-free setting, when the coefficient of variation is unknown. We make use of the next lower bound for Quantum Counting.

\begin{proposition}[Theorem 4.2.6 in \cite{Nay99d}] 
  \label{Prop:QCountLower}
  Let $N > 0$, $1 < K \leq N$ and $\eps \in \pt[\big]{\frac{1}{4K},1}$. Consider the set of all quantum algorithms such that, given a query oracle to any function $f : [N] \ra \rn$, they return an estimate $\td{C}$ of the number $C$ of preimages of $1$ in $f$ such that $\abs{\td{C} - C} \leq \eps C$ with probability at least~$2/3$.
  Let $T_K(N,\eps)$ be the minimum number of quantum queries any such algorithm must use when the oracle has exactly~$K$ preimages of~$1$. Then, $T_K(N,\eps) \geq \om[\Big]{\frac{\sqrt{K(N-K)}}{\eps K+1} + \sqrt{\frac{N}{\eps K+1}}}$.
\end{proposition}

We obtain by a simple reduction to the above problem that the result described in Theorem~\ref{Thm:SeqEstim} is nearly optimal.

\begin{proposition}
  \label{Prop:epsdeltaFreeLower}
  Let $\eps \in (0,1)$. Let $\mathcal{P}_{\mathcal{B}}$ denote the family of all \qrvs\ that follow a Bernoulli distribution. Consider any algorithm that takes as input $X \in \mathcal{P}_{\mathcal{B}}$ and that outputs a mean estimate $\mut$ such that $\abs{\mut - \ex{X}} \leq \eps \ex{X}$ with probability at least~$2/3$. Then, for any $\mu \in (0,1)$, there exists $X \in \mathcal{P}_{\mathcal{B}}$ with mean $\mu$ such that the algorithm performs at least $\om*{\frac{\sigma}{\eps \mu} + \frac{1}{\sqrt{\eps \mu}}}$ quantum experiments on input $X$, where $\sigma^2 = \var{X}$.
\end{proposition}

\begin{proof}
  Given $\eps \in (0,1)$ and $\mu \in (0,1)$, we choose two integers $K$ and $N$ such that $K > 1/(4\eps)$ and $K/N = \mu$ (assuming $\mu$ is rational). Similarly to the proof of Theorem~\ref{Thm:SubGLower}, we associate with each function $f : [N] \ra \rn$ the q-variable $(\Hil,U,M)_{f}$ where $\Hil = \C^{N+2}$, $U \qub = \frac{1}{\sqrt{N}} \sum_{x \in [N]} \ket{x}\ket{f(x)}$, and $M = \set{I \otimes \proj{0}, I \otimes \proj{1}}$. If a quantum algorithm can estimate the mean of any Bernoulli random variable with error $\eps$ and success probability~$2/3$, then it can be used to count the number of preimages of $1$ in $f$ with the same accuracy. Thus, by Proposition~\ref{Prop:QCountLower}, any such algorithm must perform at least
    $\om[\Big]{\frac{\sqrt{K(N-K)}}{\eps K+1} + \sqrt{\frac{N}{\eps K+1}}}
        = \om[\Big]{\frac{\sqrt{\mu(1-\mu)}}{\eps \mu + 1/N} + \frac{1}{\sqrt{\eps \mu + 1/N}}}
        = \om*{\frac{\sigma}{\eps \mu} + \frac{1}{\sqrt{\eps \mu}}}$
  quantum experiments on a \qrv\ with mean $\mu$ and variance~$\sigma^2 = \mu(1-\mu)$.
\end{proof}


\subsection{State-based estimation}
\label{Sec:stateLower}

We consider the \emph{state-based} model where the input consists of several copies of a quantum state $\ket{p} = \sum_{x \in E} \sqrt{p(x)} \ket{x}$ encoding a distribution~$p$ over $E$. This model is weaker than the one described before, since it does not provide access to a unitary algorithm preparing~$\ket{p}$. We prove that no quantum speedup is achievable in this setting. Our result uses the next lower bound on the number of copies needed to distinguish two states.

\begin{lemma}
  \label{Lem:KLDiv}
  Let $\delta \in (0,1)$ and consider two probability distributions $p_0$ and $p_1$ with the same finite support $E$. Define the states $\ket{\phi_0} = \sum_{x \in E} \sqrt{p_0(x)}\ket{x}$ and $\ket{\phi_1} = \sum_{x \in E} \sqrt{p_1(x)}\ket{x}$. Then, the smallest integer $T$ such that there is an algorithm that can distinguish $\ket{\phi_0}^{\otimes T}$ from~$\ket{\phi_1}^{\otimes T}$ with success probability at least $1-\delta$ satisfies $T \geq \frac{\ln(1/(4\delta))}{D(p_0 \| p_1)}$, where $D(p_0 \| p_1) = \sum_{x \in E} p_0(x) \ln\pt*{\frac{p_0(x)}{p_1(x)}}$ is the KL-divergence from $p_0$ to~$p_1$.
\end{lemma}

\begin{proof}
  According to Helstrom's bound \cite{Hel69j} the best success probability to distinguish between two states $\ket{\phi}$ and $\ket{\phi'}$ is $\frac{1}{2}(1+\sqrt{1-\abs{\ip{\phi}{\phi'}}^2})$. Thus, the smallest number $T$ needed to distinguish $\ket{\phi_0}^{\otimes T}$ from $\ket{\phi_1}^{\otimes T}$ must satisfy $\frac{1}{2}(1+\sqrt{1-\ip{\phi_0}{\phi_1}^{2T}}) \geq 1-\delta$. It implies that
    $T \geq \frac{- \ln\pt*{1 - (1-2\delta)^2}}{-2\ln\pt*{\ip{\phi_0}{\phi_1}}}
        \geq \frac{\ln(1/(4\delta))}{-2\ln\pt[\Big]{\sum\limits_{x \in E} p_0(x) \sqrt{\frac{p_1(x)}{p_0(x)}}}}
        \geq \frac{\ln(1/(4\delta))}{\sum\limits_{x \in E} p_0(x) \ln\pt*{\frac{p_0(x)}{p_1(x)}}}
        = \frac{\ln(1/(4\delta))}{D(p_0 \| p_1)}$
  where the second inequality uses the concavity of the logarithm function.
\end{proof}

We use the above lemma to show that no quantum mean estimator can perform better than the classical sub-Gaussian estimators in the state-based input model.

\begin{theorem}
  \label{Thm:StateBLower}
  Let $\nb > 1$ and $\delta \in (0,1)$ such that $\nb \geq 2\log(1/\delta)$. Fix $\sigma > 0$ and consider the family $\mathcal{P}_{\sigma}$ of all distributions with finite support whose variance lies in the interval $[\sigma^2,4\sigma^2]$. For any $p \in \mathcal{P}_{\sigma}$ with support $E \subset \R$, define the state $\ket{p} = \sum_{x \in E} \sqrt{p(x)}\ket{x}$. Let~$T(\nb,\sigma,\delta)$ be the smallest integer such that there exists an algorithm that receives the state~$\ket{p}^{\otimes T(\nb,\sigma,\delta)}$ for any $p \in \mathcal{P}_{\sigma}$, and that outputs an estimate $\mut$ of the mean $\mu$ of $p$ such that $\pr*{\abs{\mut - \mu} > \sqrt{\frac{\sigma^2 \log(1/\delta)}{\nb}}} \leq \delta$. Then,~$T(\nb,\sigma,\delta) \geq \om{\nb}$.
\end{theorem}

\begin{proof}
  Let $\mb = \frac{\nb}{\log(1/\delta)}$, $b = \frac{\mb}{\sqrt{\mb-1}}\sigma$ and $\alpha = 2 \ln\pt*{1 + \sqrt{1 - \frac{1}{\mb}}}$. We define the two  distributions $p_0$ and $p_1$ with support $E = \set{0,b}$ such that $p_0(b) = \frac{e^{\alpha}}{\mb}$ and $p_1(b) = \frac{1}{\mb}$. Let~$\mu_0$ and $\sigma_0^2$ (resp. $\mu_1$ and $\sigma_1^2$) denote the expectation and the variance of $p_0$ (resp. $p_1$). Observe that $p_0, p_1 \in \mathcal{P}_{\sigma}$ since $\sigma_0 \in [\sigma,2\sigma]$ and $\sigma_1 = \sigma$. Moreover, $\mu_0 - \mu_1 = \sigma \frac{e^{\alpha}-1}{\sqrt{\mb-1}} = \sigma \pt[\big]{e^{\alpha/2} + 1} \frac{e^{\alpha/2}-1}{\sqrt{\mb-1}} > 2 \sqrt{\frac{\sigma^2 \log(1/\delta)}{\nb}}$. Thus, we can distinguish $\ket{p_0}^{\otimes T(\nb,\sigma,\delta)}$ from $\ket{p_1}^{\otimes T(\nb,\sigma,\delta)}$ with failure probability~$\delta$ by using any optimal algorithm that satisfies the error bound stated in the theorem. Since the KL-divergence from $p_0$ to $p_1$ is $D(p_0 \| p_1) \leq p_0(b) \ln\pt*{\frac{p_0(b)}{p_1(b)}} = \frac{\alpha e^{\alpha}}{\mb^2} \leq \frac{6}{\mb}$, we must have $T(\nb,\sigma,\delta) \geq \om*{\frac{\log(1/\delta)}{D(p_1 \| p_0)}} = \om*{\nb}$ by Lemma~\ref{Lem:KLDiv}.
\end{proof}

\section{Discussion}
\label{Dis:Mean}
One interesting open question is to find a quantum mean estimator that achieves the deviation bound $\pr[\big]{\abs{\mut - \mu} > \frac{\sigma \log(1/\delta)}{\nb}} \leq \delta$ by performing a number of experiments that is \emph{linear} in $\nb$. The current best upper bound (Theorem~\ref{Thm:SubGaussian}) is $\bo{\nb \log^{3/2}(\nb) \log\log(\nb)}$, and the lower bound is $\om{\nb}$ (Theorem~\ref{Thm:SubGLower}). A first step toward this goal could be to obtain a better algorithm for the restricted case of Gaussian distributions. An equivalent goal is to find the smallest value~$L$ such that the deviation bound $\pr[\big]{\abs{\mut - \mu} > L \frac{\sigma \log(1/\delta)}{\nb}} \leq \delta$ can be achieved by a quantum mean estimator that performs \emph{at most}~$\nb$ quantum experiments. Classically, for the sub-Gaussian deviation bound of Equation~(\ref{Eq:Gaussian}), the optimal value is $L = \sqrt{2}(1 + o(1))$~\cite{Cat12j,LV20p}.

There exist many variants of the quantum mean estimation problem that have not been completely explored in the quantum model yet. Let us mention for instance the multivariate setting~\cite{LM19j}, where the objective is to estimate the mean of a random variable taking values in~$\R^d$. Heinrich~\cite{Hei04j} proved that no quantum speed-up is achievable under some condition on the largeness of~$d$. On the other hand, Cornelissen and Jerbi~\cite{CJ21p} obtained partial quantum speed-ups for some parameter settings. We also note that the first polynomial-time classical algorithm with a sub-Gaussian error rate for this problem was only found recently by Hopkins~\cite{Hop20j}.



\setlength{\emergencystretch}{5em} 
\printbibliography

\appendix

\section{Auxiliary algorithms}

  \subsection{Amplitude amplification}
  The amplitude amplification algorithm~\cite{BHMT02j} is a generalization of Quantum Searching to the problem of boosting the success probability of a quantum algorithm that performs no intermediate measurement. The next result corresponds to Equation (8) in~\cite{BHMT02j}.

\begin{theorem}[\sc Amplitude amplification, \cite{BHMT02j}]
  \label{Thm:AA}
  Let $U$ be a unitary quantum algorithm and let $\Pi$ be a projection operator. Consider the angle $\theta \in [0,\frac{\pi}{2}]$ and two unit states $\ket{\psi_0}, \ket{\psi_1}$ such that $\sin(\theta) \ket{\psi_1} = \Pi U \qub$ and $U \qub = \cos(\theta) \ket{\psi_0} + \sin(\theta) \ket{\psi_1}$. Then, for any integer $\nb \geq 0$, the \emph{amplitude amplification} algorithm $\aamp(U,\Pi,\nb)$ satisfies
    \[\aamp(U,\Pi,\nb) \qub = \cos((2\nb+1)\theta) \ket{\psi_0} + \sin((2\nb+1)\theta) \ket{\psi_1}.\]
  The algorithm uses $\nb+1$ applications of $U$, $\nb$ applications of $U^\dagger$, and $\nb$ applications of the reflection operator $I - 2\Pi$.
\end{theorem}

Next, we present a variant of the amplitude amplification algorithm that does not use a pre-defined number of applications of $U$ and $U^\dagger$. We call it the ``sequential amplitude amplification'' algorithm in reference to \emph{sequential analysis}. The original version of this algorithm was analysed in Theorem 3 of \cite{BBHT98j,BHMT02j}, with a bound on the expected time complexity $\ex{T}$. We propose a slightly different version that allows us to bound $\ex{T^2}$ and $\ex{1/T}$ (note that $\ex{T} \leq \sqrt{\ex{T^2}}$). These bounds will be used in Theorem~\ref{Thm:SeqAE} and Theorem~\ref{Thm:SeqEstim}.

\algobox{Alg:SeqAA}{Sequential amplitude amplification.}{
\begin{enumerate}
  \item Set $\ell = 0$ and $\lambda = 1.1$.
  \item Increase $\ell$ by $1$ and pick an integer $\nb$ between $\lambda^{\ell-1}$ and $\lambda^{\ell} - 1$ uniformly at random.
  \item Apply the amplitude amplification algorithm $\aamp(U,\Pi,\nb)$ (Theorem~\ref{Thm:AA}) on $\ket{0}$ and measure the state by using the projective measurement $\{I-\Pi,\Pi\}$. If the outcome is~``$\Pi$'' then stop and output the obtained state. Otherwise, go to step 2.
\end{enumerate}
}

\begin{theorem}[\sc Sequential amplitude amplification]
  \label{Thm:SeqAA}
  Let $U$ be a unitary quantum algorithm and let $\Pi$ be a projection operator. Define the number $p \in [0,1]$ and the two unit states $\ket{\psi_0}, \ket{\psi_1}$ such that $U \ket{0} = \sqrt{1-p} \ket{\psi_0} + \sqrt{p} \ket{\psi_1}$ and $\Pi U \ket{0} = \sqrt{p} \ket{\psi_1}$. If $p > 0$ then the \emph{sequential amplitude amplification} algorithm $\saamp(U,\Pi)$ outputs the state $\ket{\psi_1}$ with probability $1$. Moreover, if we let $T$ denote the number of applications of $U$, $U^\dagger$ and $I-2\Pi$ used by the algorithm, then $\ex{T^2} \leq \bo{1/p}$ and $\ex{1/T} \leq \bo{\sqrt{p}}$.
\end{theorem}

\begin{proof}
  Let $0 < \theta \leq \pi/2$ be the angle such that $\sqrt{p} = \sin \theta$.
  We show the theorem in the case where $\theta < \pi/4$ (the case $\theta \geq \pi/4$ is easy to handle separately). We first prove the inequality $\ex{T^2} \leq \bo{1/p}$. Let $P_{\ell}$ denote the probability of obtaining $\ket{\psi_1}$ when $\nb$ is picked uniformly at random between $\lambda^{\ell-1}$ and $\lambda^{\ell} - 1$ and the state $\mathcal{A}(U,\Pi,\nb) \ket{0}$ is measured with respect to $\{I-\Pi,\Pi\}$. Let $\ell_+ = \left\lceil \log_{\lambda}\pt*{\frac{\lambda}{(\lambda-1)\sin(2\theta)}} \right\rceil$. If $\ell \geq \ell_+$ then,
    \begin{align*}
     P_{\ell} & = \frac{1}{(\lambda - 1)\lambda^{\ell-1}} \sum_{\nb = \lambda^{\ell-1}}^{\lambda^{\ell} - 1} \sin^2((2\nb+1)\theta) \tag*{by Theorem~\ref{Thm:AA}} \\
       & = \frac{1}{(\lambda - 1)\lambda^{\ell-1}} \sum_{\nb = \lambda^{\ell-1}}^{\lambda^{\ell} - 1} \frac{1 - \cos^2((2\nb+1)2\theta)}{2} \tag*{by a trigonometric identity} \\
       & \geq \frac{1}{2} - \frac{1}{2(\lambda - 1)\lambda^{\ell-1}} \sum_{\nb = 0}^{\lambda^{\ell} - 1} \cos^2((2\nb+1)2\theta) \\
       & = \frac{1}{2} - \frac{\sin(4\lambda^{\ell}\theta)}{4(\lambda-1)\lambda^{\ell-1}\sin(2\theta)} \tag*{by a trigonometric identity} \\
       & \geq \frac{1}{4} \tag*{by $\ell \geq \left\lceil \log_{\lambda}\pt*{\frac{\lambda}{(\lambda-1)\sin(2\theta)}} \right\rceil$}
    \end{align*}
  Moreover, the algorithm has used at most $\sum_{\nb = 1}^{\ell} \lambda^{\nb} \leq 10 \lambda^{\ell+1}$ applications of $U$, $U^\dagger$ and $I-2\Pi$ after~$\ell$ iterations of step 3. Consequently, $\ex{T^2} \leq \sum_{\ell \geq \ell_+} \pt{10 \lambda^{\ell+1}}^2 (3/4)^{\ell - \ell_+} \leq \bo{\lambda^{2\ell_+}} \leq \bo{1/p}$.

  Next, we prove the inequality $\ex{1/T} \leq \bo{\sqrt{p}}$. Let $\ell_- = \left\lceil \log_{\lambda}\pt*{\frac{1}{5\theta}} \right\rceil$. We will use that $(2/\pi)x \leq \sin(x) \leq x$ for all $x \in [0,\pi/2]$. According to Theorem~\ref{Thm:AA}, if $\nb$ is an integer between $\lambda^{\ell-1}$ and $\lambda^{\ell} - 1$ for some $1 \leq \ell \leq \ell_-$, then the probability of obtaining $\ket{\psi_1}$ when $\mathcal{A}(U,\Pi,\nb) \ket{0}$ is measured is $\sin^2((2\nb+1)\theta) \leq (2 \lambda^\ell-1)^2\theta^2$. Moreover, after~$\ell$ iterations of step 3, the algorithm used at least $\lambda^{\ell-1}$ applications of $U$, $U^\dagger$, $I-2\Pi$ in total. Consequently, $\ex{1/T} \leq \sum_{\ell = 1}^{\ell_-} \frac{(2 \lambda^\ell-1)^2\theta^2}{\lambda^{\ell-1}} + \frac{1}{\lambda^{\ell_-}} \leq \bo{\lambda^{\ell_-} \theta^2 + \theta} \leq \bo{\sqrt{p}}$.
\end{proof}

  \subsection{Amplitude estimation}
  The amplitude estimation algorithm~\cite{BHMT02j} is a generalization of Quantum Counting to the problem of estimating the success probability of an algorithm. The next result corresponds to Theorems 11 and 12 in~\cite{BHMT02j}.

\begin{theorem}[\sc Amplitude estimation, \cite{BHMT02j}]
  \label{Thm:AE}
  Let $U$ be a unitary quantum algorithm and let $\Pi$ be a projection operator. Define the number $p \in [0,1]$ such that $p = \norm{\Pi U \qub}^2$.
  Then, for any integer $\nb \geq 0$, the \emph{amplitude estimation} algorithm $\aest(U,\Pi,\nb)$ outputs an amplitude estimate $\td{p}$ such that,
    \[\pr*{\abs*{\td{p} - p} \leq \frac{2\pi \sqrt{p(1-p)}}{\nb} + \frac{\pi^2}{\nb^2}} \geq 8/\pi^2.\]
  The algorithm uses $\nb$ applications of $U$, $U^\dagger$, $I-2\Pi$ and $\bo{\log^2(\nb)}$ $2$-qubit quantum gates.
\end{theorem}

We present a sequential version of the amplitude estimation algorithm that does not need a time parameter $\nb$ as input. This result was first obtained by \cite[Theorem 15]{BHMT02j}. We describe a variant with additional properties that is based on the sequential amplitude amplification algorithm. It is used in Proposition~\ref{Prop:SeqZeroOne} and Theorem~\ref{Thm:SeqEstim}.

\begin{theorem}[\sc Sequential amplitude estimation]
  \label{Thm:SeqAE}
  There exists an algorithm, called the \emph{sequential amplitude estimation} algorithm $\saest$, with the following properties.
  Let $U$ be a unitary quantum algorithm and let $\Pi$ be a projection operator. Define the number $p \in [0,1]$ such that $p = \norm{\Pi U \ket{0}}^2$.
  Then, the algorithm $\saest(U,\Pi)$ outputs an amplitude estimate $\td{p}$ and uses a number $T$ of applications of $U$, $U^\dagger$, $I-2\Pi$ such that,
    \begin{enumerate}
      \item There is a universal constant $c \in (0,1)$ such that $\pr{\abs{\td{p} - p} \leq c p} \geq 7/8$.
      \item There is a universal constant $c'$ such that $\ex{T^2} = \ex{1/\td{p}} \leq c'/p$.
      \item There is a universal constant $c''$ such that $\ex{1/T} = \ex{\sqrt{\td{p}}} \leq c''\sqrt{p}$.
    \end{enumerate}
\end{theorem}

\begin{proof}
  The algorithm $\saest(U,\Pi)$ consists of recording the number $T$ of applications of $U$, $U^\dagger$, $I-2\Pi$ used by the sequential amplitude amplification algorithm $\saamp(U,\Pi)$ (Theorem~\ref{Thm:SeqAA}), and choosing the estimate $\td{p} = 1/T^2$. The results follow immediately from Theorem~\ref{Thm:SeqAA} and Markov's inequality.
\end{proof}


\end{document}